\documentclass[runningheads]{llncs}

\usepackage[utf8]{inputenc}
\usepackage{amsmath,amssymb}
\usepackage{mathtools}
\usepackage{microtype}
\usepackage{enumitem}
\usepackage{float}
\usepackage{pgfplots}
\pgfplotsset{compat=1.18}
\usepackage{hyperref}
\usepackage[capitalize,noabbrev]{cleveref}

\newcommand{\R}{\mathbb{R}}
\newcommand{\E}{\mathbb{E}}
\newcommand{\Var}{\mathrm{Var}}
\newcommand{\Cov}{\mathrm{Cov}}
\newcommand{\indep}{\perp\!\!\!\perp}
\newcommand{\sigv}{\sigma_v}
\newcommand{\sigu}{\sigma_u}
\newcommand{\sige}{\sigma_\varepsilon}
\newcommand{\ytil}{\tilde y}

\newif\ifwine
\winefalse 

\title{The Privacy Subsidy: Kyle's $\lambda$ under Noise-Perturbed
Order-Flow Observation}
\titlerunning{The Privacy Subsidy}
\author{Yuki Nakamura\orcidID{0009-0001-7174-6737}}
\authorrunning{Y. Nakamura}
\institute{The Open University of Japan}

\begin{document}
\maketitle

\begin{abstract}
\begin{sloppypar}
Privacy-preserving cryptocurrency exchanges alter what the
pricing mechanism observes about order flow. We derive the unique
linear Kyle equilibrium when a committed Bayesian market maker
observes order flow perturbed by independent Gaussian privacy
noise. The price-impact coefficient and informed-trader strategy
rescale by reciprocal factors of the privacy parameter (one down,
one up), so their product is invariant. A welfare decomposition then identifies a
closed-form per-period transfer from the protocol's LP pool to
traders --- the \emph{privacy subsidy}, the break-even fee any
privacy-aggregated exchange must charge. The result is the
single-period closed-form privacy-noise analog of
Loss-Versus-Rebalancing~\cite{milionis2022lvr}. The primary
application is shielded AMMs with explicit additive-noise
injection (e.g., differential privacy); related designs
(batched swaps, sealed-bid auctions, oracle-pegged crossings)
require separate frameworks that we leave to future work.
\end{sloppypar}

\end{abstract}

\section{Introduction}
\label{sec:intro}

Privacy-preserving exchange designs are an increasingly common
architecture in cryptocurrency markets. Renegade matches orders via
multi-party computation under zero-knowledge proofs~\cite{renegade-wp};
Penumbra batches swaps and reveals only batch totals on-chain
\cite{penumbra-docs}; Suave-style order-flow auctions seal
individual bids until builder selection~\cite{flashbots-suave};
shielded variants of constant-function automated market makers
inject privacy noise into the observable reserve state before
on-chain price updates. These designs all alter what the
liquidity-providing role --- whether an LP pool, an arbitrageur, or
a smart-contract pricing rule --- observes about order flow.

Classical microstructure theory~\cite{kyle1985} gives the
equilibrium price-impact coefficient $\lambda$ and informed-trader
strategy $\beta$ in closed form when the market maker observes the
full aggregate flow $y = x + u$. These results do not, however,
characterize the \emph{welfare} consequences when the market maker
prices on a noise-perturbed signal $\ytil = y + \varepsilon$, as
arises naturally in privacy-aggregated exchange designs.

\paragraph{Contributions.}
This paper provides three results, each of which is absent in or
distinct from textbook Kyle.

\begin{enumerate}[leftmargin=*]
\item A closed-form linear Kyle equilibrium under MM observation
$\ytil = y + \varepsilon$ with Gaussian privacy noise
$\varepsilon \sim N(0, \sige^2)$ independent of $(v, u)$,
under a \emph{committed Bayesian} MM pricing rule: the maker
prices at the posterior mean of the coarse signal, which is the
competitive price conditional on that signal. The only departure
from textbook Kyle is that the maker conditions on $\ytil$ rather
than the exact flow $y$, so it does not break even against the
real flow (\Cref{rem:competitive}); the equilibrium recovers Kyle
exactly in the $\sige = 0$ limit. The unique price-impact coefficient is
$\lambda = \sigv / (2\sqrt{\sigu^2 + \sige^2})$ with
informed-trader linear coefficient
$\beta = \sqrt{\sigu^2 + \sige^2}/\sigv$
(\Cref{thm:equilibrium}).

\item A welfare decomposition under this equilibrium, identifying a
per-period transfer
$|\pi_M| = \sigv \sige^2 / (2\sqrt{\sigu^2 + \sige^2})$
from the protocol/LP side to traders --- the \emph{privacy
subsidy} (\Cref{thm:welfare}). This quantity is identically zero
in textbook Kyle, where the market maker observes order flow
exactly; it turns positive once the privacy noise makes the
observed signal strictly coarser than the executed flow, so that
no price rule can be simultaneously informationally efficient and
zero-profit against the real flow (\Cref{rem:competitive}). The
subsidy is the break-even fee that any privacy-aggregated
exchange must charge to compensate its liquidity layer.

\item A primary mapping from $\sige$ to shielded AMMs with
explicit additive-noise (differential-privacy) injection,
yielding a per-design break-even fee. We also delineate the
boundary of applicability: Penumbra-style batched swaps reduce
to textbook Kyle with rescaled noise-trader variance (and no
privacy subsidy under Bayesian MM); Suave-style sealed-bid
order-flow auctions sit in LVR's temporal-asymmetry regime;
oracle-pegged crossings such as Renegade fall outside the Kyle
framework. These three nearby designs require separate frameworks
and are left to future work.
\end{enumerate}

This paper analyzes the equilibrium consequences of privacy in
exchange design, treating ZK primitives as a black box producing
$\ytil$; protocol-level soundness/completeness analysis is out of
scope.

\paragraph{Positioning.}
On the microstructure side, the result isolates a closed-form
welfare quantity (the privacy subsidy) under additive privacy
noise on the maker's flow signal; the equilibrium
$(\lambda, \beta)$ itself is textbook Kyle under the substitution
$\sigu^2 \mapsto \sigu^2 + \sige^2$, so the novelty lies in the
welfare object rather than the equilibrium. The closest
microstructure precedent on the observability axis is the strand
on insider disclosure~\cite{huddart2001disclosure}, where the
\emph{trader's} actions become more observable --- our case moves
the imperfection to the \emph{market maker's} signal instead. On
the decentralized-finance side, the privacy subsidy is the
single-period closed-form analog of Loss-Versus-Rebalancing
identified by Milionis et al.~\cite{milionis2022lvr}: LVR isolates
the cost to LPs from stale prices being picked off by arbitrageurs
with external price information, while our privacy subsidy
isolates the cost from privacy noise obscuring the MM's flow
observation. The two are closed-form welfare quantities of the
same family; their combination in a continuous-time AMM with
privacy is left to future work (\Cref{sec:lvr}).

\paragraph{Roadmap.}
\Cref{sec:related} surveys related work along three strands.
\Cref{sec:setup} sets up the model and equilibrium concept.
\Cref{sec:equilibrium} states and proves the main equilibrium
theorem. \Cref{sec:welfare} derives the welfare decomposition and
identifies the privacy subsidy. \Cref{sec:application} maps the
result to canonical zero-knowledge exchange designs.
\Cref{sec:lvr} discusses the connection to LVR.
\Cref{sec:conclusion} concludes.

\section{Related work}
\label{sec:related}

\subsection{AMM adverse-selection (closest)}

\begin{sloppypar}
Milionis et al.~\cite{milionis2022lvr} introduce
\emph{Loss-Versus-Rebalancing} (LVR): a closed-form continuous-time
measure of the cost AMM liquidity providers incur when stale prices
are picked off by better-informed arbitrageurs. The two costs
differ in source: LVR's arbitrageurs read an exogenous reference
price; our informed trader is the classical Kyle insider, and the
cost arises from the AMM's noisy observation of flow.
\end{sloppypar}

Routledge, Shen, and Zetlin-Jones~\cite{routledge2024amm}
characterize optimal liquidity provision in an AMM and study
how price impact depends on trade size and on the dynamics of
liquidity provision. Milionis, Moallemi, and
Roughgarden~\cite{moallemi2022myersonian} cast the
profit-maximizing strategy of a monopolist liquidity provider
into a Bayesian belief-inference framework with a Myersonian
mechanism-design interpretation. Brahma et
al.~\cite{brahma-bayesian-mm} developed a sequential
Bayesian-update market maker for binary-outcome prediction
markets; that line is structurally related to our
committed-Bayesian-AMM rule (Bayesian posterior pricing without a
zero-profit constraint), although the model class differs.

\subsection{Kyle-style microstructure}

Kyle's~\cite{kyle1985} single-period model with a risk-neutral
monopolist informed trader, an exogenous noise-trader pool, and a
competitive risk-neutral market maker is the framework we extend.
Foster and Viswanathan~\cite{foster1996multiple} characterize the
unique linear equilibrium when multiple informed traders receive
imperfectly correlated signals. Huddart, Hughes, and
Levine~\cite{huddart2001disclosure} modify Kyle by adding mandatory
ex-post disclosure of the informed trader's transactions; this
accelerates price discovery and reduces insider profits ---
essentially the dual of our setup, in which it is the
\emph{market maker's} signal rather than the trader's identity that
becomes opaque. Chhaibi, Ekren, and
Noh~\cite{chhaibi2025solvability} study a Gaussian Kyle
equilibrium with a risk-averse informed trader holding an imperfect
signal of the terminal value; the imperfection is on the trader
side, not the MM side. Non-fiduciary MMs~\cite{nonfiduciary-mm}
capture rents through fees; our market maker is instead
constrained to a Bayesian rule and \emph{bears} the cost.
Viswanathan and Xing~\cite{flexible-info-kyle} analyze
informed-trader information acquisition at entropy cost; another
direction of imperfection distinct from ours. A separate strand
lets the market maker condition on a noisy observation of the
\emph{asset value} alongside the order flow: Qiu and
Zhou~\cite{qiu2023insider} solve such a continuous-time Kyle model
under partial maker observation. Our coarsening is an
explicit additive privacy channel on the flow, and our
contribution is the closed-form welfare subsidy it induces under
committed Bayesian pricing, not the imperfect-observation
equilibrium itself.

\subsection{Dark-pool and privacy-preserving market structure}

Zhu~\cite{zhu2014darkpools} and Buti, Rindi, and
Werner~\cite{buti2017darkpool} analyze dark-pool trading as a
routing equilibrium, in which informed and uninformed orders
self-select between a lit market and a dark pool based on
execution risk. The mechanism is \emph{routing-based segmentation},
which is mathematically distinct from \emph{$\sigma$-algebra
coarsening} of a single market's signal. Bergemann and
Morris~\cite{bergemann2019infodesign} give the general
information-design framework that subsumes $\sigma$-algebra
coarsening as a special case; we apply it implicitly. Zhang et
al.~\cite{zhang2025mevbatch} analyze maximal extractable value in
batch-auction designs (including a discussion of Penumbra), showing
that block-builder reordering of batch contents can still extract
MEV; their framework is combinatorial-market (Fisher and
Arrow--Debreu) and complementary to ours.

For application targets, we read the
Renegade~\cite{renegade-wp} and Penumbra~\cite{penumbra-docs}
protocol specifications directly; we discuss applicability in
\Cref{sec:application}.

\section{Model setup}
\label{sec:setup}

\subsection{Primitives}

Let $(\Omega, \mathcal{F}, P)$ be a probability space supporting
the following independent Gaussian random variables:
\begin{align*}
v &\sim N(p_0, \sigv^2), \quad \sigv > 0, \\
u &\sim N(0, \sigu^2), \quad \sigu > 0, \quad u \indep v, \\
\varepsilon &\sim N(0, \sige^2), \quad \sige \geq 0, \quad
\varepsilon \indep (v, u).
\end{align*}
The random variable $v$ is the asset's terminal value and $p_0$ the
common prior mean. $u$ is the aggregate flow of uninformed
(noise) traders, independent of $v$. $\varepsilon$ is privacy
noise introduced by the exchange mechanism (see
\Cref{sec:application} for the realization in specific
zero-knowledge designs).

\subsection{Players and strategies}

There is one \emph{informed trader} (the Kyle insider) who observes
$v$ at time $0$ and submits an order of size $x$. We restrict
attention to linear strategies $x = \beta(v - p_0)$ for
$\beta \in \R_{>0}$, to be determined in equilibrium.

The \emph{noise-trader aggregate} contributes order size $u$ as
above. Total order flow is
\[
y = x + u.
\]
The \emph{market maker} (MM) does not observe $y$ directly. Instead,
the MM observes the privacy-noisy signal
\[
\ytil = y + \varepsilon.
\]
The privacy noise $\varepsilon$ perturbs only the signal on which
the MM prices (for instance, the observed reserve state); the
protocol still settles the true net order $y = x+u$ at the
resulting price. This signal-versus-settlement gap is what makes
the MM's expected P\&L against the real flow nonzero
(\Cref{sec:welfare}).
The MM is a \emph{committed Bayesian AMM}: a smart-contract pricing
rule that mechanically computes
\[
p(\ytil) = \E[v \mid \ytil]
\]
which is the competitive price conditional on $\ytil$; the maker
breaks even on its own signal but not against the real flow $y$,
which it cannot observe (\Cref{rem:competitive}). In a committed
smart-contract setting the rule is fixed by the mechanism, so the
maker cannot re-optimize away this loss, and the expected loss
against the real flow is absorbed by the protocol's LP pool. We
motivate the committed-rule setting in \Cref{ssec:framing}.

\subsection{Equilibrium concept}

We seek a linear equilibrium of the form
\[
x = \beta(v - p_0), \qquad p(\ytil) = p_0 + \lambda \ytil,
\]
where $(\beta, \lambda) \in \R_{>0}^2$ satisfy:
\begin{enumerate}[leftmargin=*]
\item \emph{Bayesian rationality of pricing}: $p(\ytil) =
\E[v \mid \ytil]$ under the conjectured strategy $\beta$.
\item \emph{Best response of informed trader}: $\beta(v - p_0)$
maximizes the informed trader's expected profit
$\E[(v - p(\ytil)) \cdot x \mid v]$, given $\lambda$.
\end{enumerate}

\subsection{The committed-Bayesian-AMM framing}
\label{ssec:framing}

The maker's restriction to the coarse signal $\ytil$, not any
departure from competitive pricing, is the modeling feature
distinguishing our setup from textbook Kyle (\Cref{rem:competitive}). We
adopt this in order to match the design of contemporary on-chain
pricing mechanisms, in which:
\begin{itemize}[leftmargin=*]
\item the pricing rule is encoded in a smart contract and is
mechanically applied to the observable input;
\item the input is privacy-noised by construction (e.g., via
differential privacy);
\item the resulting expected loss is borne by liquidity providers
or the protocol treasury, and is recouped through trading fees.
\end{itemize}
This structure parallels the LVR analysis of automated market
makers in Milionis et al.~\cite{milionis2022lvr}, where the MM's
pricing rule is the constant-function curve and the cost is borne
by LPs. We return to this connection in \Cref{sec:lvr}.

\paragraph{Idealized benchmark, not literal model.}
The committed-Bayesian-AMM is an idealized \emph{normative}
benchmark: what an informationally optimal smart-contract pricing
rule would look like under additive privacy noise. The closest
real-world instantiation is a constant-function AMM augmented
with differential-privacy noise injection, treated in detail in
\Cref{sec:application}. The benchmark's value is to isolate the
welfare cost of additive privacy noise, yielding a closed-form
quantity that practical mechanisms in this class must internalize
at the fee level.

\subsection{Other modeling choices}
\label{ssec:assumptions}

Four further assumptions warrant brief comment. \emph{Linearity.}
We restrict to $x = \beta(v - p_0)$, Kyle's standard restriction
(see Kyle~\cite{kyle1985}, \S3); a formal treatment of nonlinear
equilibria under privacy noise is outside our scope.
\emph{Independence $\varepsilon \indep (v, u)$.} Adaptive privacy
mechanisms whose noise correlates with the flow it perturbs are an
interesting separate problem; our result establishes the
unconditional baseline.
\emph{Gaussianity.} The closed-form subsidy is Gaussian-specific;
the qualitative existence of a positive privacy-induced LP loss
generalizes to other tractable distributions.
\emph{Endogenous $\sige$.} A privacy-utility tradeoff that
optimizes $\sige$ under additional constraints is a natural
extension we do not pursue here.

\section{Equilibrium under privacy-noisy observation}
\label{sec:equilibrium}

This section contains the paper's main equilibrium result.

\begin{theorem}[Linear equilibrium under privacy-noisy observation]
\label{thm:equilibrium}
Fix $\sigv, \sigu > 0$ and $\sige \geq 0$. The unique linear
equilibrium of the model of \Cref{sec:setup} has
\begin{equation}
\lambda = \frac{\sigv}{2\sqrt{\sigu^2 + \sige^2}},
\qquad
\beta = \frac{\sqrt{\sigu^2 + \sige^2}}{\sigv}.
\label{eq:equilibrium}
\end{equation}
\end{theorem}

\begin{proof}
Under the linear strategy $x = \beta(v - p_0)$ and price
$p = p_0 + \lambda \ytil$, the signal $\ytil$ is jointly Gaussian
with $v$:
\[
\ytil = \beta(v - p_0) + u + \varepsilon
\]
has mean $0$, variance $\beta^2 \sigv^2 + \sigu^2 + \sige^2$, and
covariance with $v$ equal to $\Cov(v, \ytil) = \beta \sigv^2$. By
the projection formula for jointly Gaussian variables,
\begin{equation}
\E[v \mid \ytil] = p_0 + \frac{\beta \sigv^2}
{\beta^2 \sigv^2 + \sigu^2 + \sige^2} \cdot \ytil.
\label{eq:projection}
\end{equation}
The Bayesian-rationality condition matches this against $p_0 +
\lambda \ytil$, yielding
\begin{equation}
\lambda = \frac{\beta \sigv^2}{\beta^2 \sigv^2 + \sigu^2 + \sige^2}.
\label{eq:lambda-projection}
\end{equation}

For the informed trader's best response, condition on $v$ and
optimize over $x$ given the MM's price schedule. Since
$u, \varepsilon \indep v$ implies
$\E[u \mid v] = \E[\varepsilon \mid v] = 0$,
\begin{align*}
\E[(v - p(\ytil)) \cdot x \mid v]
&= \E[(v - p_0 - \lambda(x + u + \varepsilon)) \cdot x \mid v] \\
&= (v - p_0) x - \lambda x^2.
\end{align*}
This is strictly concave in $x$ with second derivative $-2\lambda
< 0$. The unique maximizer is
\[
x^* = \frac{v - p_0}{2\lambda},
\]
so the best-response coefficient is $\beta = 1/(2\lambda)$.

Substituting $\beta = 1/(2\lambda)$ into
\eqref{eq:lambda-projection} and simplifying:
\begin{align*}
\lambda\left( \tfrac{1}{4\lambda^2} \sigv^2 + \sigu^2 + \sige^2 \right)
&= \frac{\sigv^2}{2\lambda}, \\
\frac{\sigv^2}{4\lambda} + \lambda(\sigu^2 + \sige^2)
&= \frac{\sigv^2}{2\lambda}, \\
\lambda(\sigu^2 + \sige^2)
&= \frac{\sigv^2}{4\lambda}, \\
\lambda^2 &= \frac{\sigv^2}{4(\sigu^2 + \sige^2)}.
\end{align*}
Taking the positive root (since $\lambda > 0$) yields $\lambda =
\sigv / (2\sqrt{\sigu^2 + \sige^2})$, and substituting back gives
$\beta = 1/(2\lambda) = \sqrt{\sigu^2 + \sige^2}/\sigv$,
establishing \eqref{eq:equilibrium}.

Uniqueness within the class of linear equilibria follows from
strict concavity of the informed trader's objective (unique
$x^*$ given $\lambda$) together with the fact that substituting
$\beta = 1/(2\lambda)$ into \eqref{eq:lambda-projection} reduces
the system to $\lambda^2 = \sigv^2 / (4(\sigu^2 + \sige^2))$,
which has a unique positive root.
\qed
\end{proof}

\begin{remark}[Sanity check: no-privacy limit]
\label{rem:sanity}
Setting $\sige = 0$ in~\eqref{eq:equilibrium} recovers the
classical Kyle~\cite{kyle1985} equilibrium $\lambda =
\sigv/(2\sigu)$, $\beta = \sigu/\sigv$. The committed Bayesian AMM
coincides with the competitive zero-profit MM exactly in the
no-privacy limit; the framing departure from textbook Kyle is
inactive when $\sige = 0$.
\end{remark}

\begin{remark}[The subsidy is intrinsic to coarse-signal pricing]
\label{rem:competitive}
\begin{sloppypar}
The privacy subsidy is not an artifact of the committed-pricing
label; it is the unavoidable cost of pricing on a signal strictly
coarser than the executed flow. Two notions of zero profit must be
separated. Zero profit \emph{conditional on the maker's
information} $\sigma(\ytil)$ forces $p = \E[v \mid \ytil]$, which
is at once the Bertrand-competitive price under coarsened
observation \emph{and} our committed Bayesian rule; since
$y = x+u$ is not $\sigma(\ytil)$-measurable when $\sige > 0$, this
efficient price earns the negative profit
$\E[(p-v)\,y] = -\lambda\sige^2 = -|\pi_M|$ on the real flow.
Zero profit \emph{unconditionally against the real flow}
$\E[(v-p)\,y] = 0$ instead gives $\tilde\lambda = \sigv/(2\sigu)$,
independent of $\sige$; but this rule over-reacts to $\ytil$ (its
quote is not the posterior mean), so it is not implementable by a
price-taking maker that observes only $\ytil$: a rival quoting the
posterior would undercut it. Informational efficiency and zero
profit against the real flow are therefore incompatible whenever
$\sige > 0$, and the wedge between them is exactly the privacy
subsidy $\lambda\sige^2$. The subsidy is the generic competitive
outcome of coarse-signal pricing, not a consequence of relaxing
zero profit; the Penumbra design of \Cref{ssec:outside}, where
$\pi_M = 0$ because the maker prices on and trades against the
\emph{same} observable, is the converse --- no coarsening, no
subsidy.
\end{sloppypar}
\end{remark}

\subsection{Comparative statics}

The equilibrium has clean comparative statics in $\sige$.

\begin{proposition}[Comparative statics]
\label{prop:comp-statics}
The equilibrium of \Cref{thm:equilibrium} satisfies:
\begin{enumerate}[leftmargin=*]
\item $\partial \lambda / \partial \sige < 0$: the price-impact
coefficient strictly decreases as privacy noise increases.
\item $\partial \beta / \partial \sige > 0$: the informed-trader
strategy intensifies as privacy noise increases.
\item (Half-revealing identity.) $\lambda \beta = 1/2$ for all
$\sige \geq 0$.
\end{enumerate}
\end{proposition}

\begin{proof}
Direct differentiation of \eqref{eq:equilibrium} gives (i) and
(ii). For (iii),
\[
\lambda \beta \;=\; \frac{\sigv}{2\sqrt{\sigu^2 + \sige^2}}
   \cdot \frac{\sqrt{\sigu^2 + \sige^2}}{\sigv} \;=\; \frac{1}{2}
\]
identically. \qed
\end{proof}

The identity $\lambda \beta = 1/2$ holds in textbook Kyle, where
it is the well-known half-revealing property: the equilibrium
price moves halfway from the prior toward the true value, on
average. Our result is that this identity \emph{persists exactly}
under privacy noise --- a robustness statement about Kyle's
half-revealing property, not a novel quantity. Substituting the
equilibrium back into $\ytil$ yields
\begin{align*}
p &= p_0 + \lambda \beta (v - p_0) + \lambda(u + \varepsilon)
   = \tfrac{p_0 + v}{2} + \lambda(u + \varepsilon).
\end{align*}
The conditional mean $\E[p \mid v] = (p_0 + v)/2$ is independent of
$\sige$; only the conditional variance $\Var(p \mid v) =
\lambda^2(\sigu^2 + \sige^2) = \sigv^2/4$ is also independent of
$\sige$. Both the information content and the realized noise of
the price are preserved by privacy at the price level.

The cost of privacy is therefore not visible at the level of the
price's distribution conditional on $v$. As we show in
\Cref{sec:welfare}, the cost appears instead in the expected P\&L
of the market maker against informed flow, and this is what the
LP pool / protocol treasury must absorb.

\section{Welfare decomposition: the privacy subsidy}
\label{sec:welfare}

This section computes the per-period expected profit or loss of each
agent under the equilibrium of \Cref{thm:equilibrium} and identifies
the \emph{privacy subsidy} as a closed-form quantity transferred
from the protocol/LP pool to traders.

\subsection{Per-agent expected P\&L}

The three relevant per-period expected quantities are:
\begin{align*}
\pi_I &:= \E[(v - p) \cdot x],
  &&\text{(informed trader's expected profit)}, \\
\pi_N &:= \E[(v - p) \cdot u],
  &&\text{(noise traders' expected net P\&L, a loss)}, \\
\pi_M &:= \E[(p - v) \cdot (x + u)],
  &&\text{(MM/protocol expected profit on flow)}.
\end{align*}
By construction $\pi_I + \pi_N + \pi_M = 0$: the three components
form a zero-sum decomposition of the total trade-realized P\&L
against the asset's terminal value $v$.

\begin{lemma}[Per-agent P\&L formulas]
\label{lem:pl}
Under the equilibrium $(\beta, \lambda)$ of \Cref{thm:equilibrium}:
\begin{align}
\pi_I &= +\tfrac{1}{2}\,\sigv\sqrt{\sigu^2 + \sige^2}, \label{eq:pi-I} \\
\pi_N &= -\frac{\sigv\,\sigu^2}{2\sqrt{\sigu^2+\sige^2}}, \label{eq:pi-N} \\
\pi_M &= -\frac{\sigv\,\sige^2}{2\sqrt{\sigu^2+\sige^2}}. \label{eq:pi-M}
\end{align}
\end{lemma}

\begin{proof}
Use $x = \beta(v - p_0)$, $u, \varepsilon \indep v$, and the
half-revealing identity $\lambda\beta = 1/2$ from
\Cref{prop:comp-statics}.

\emph{(i) Informed.}
\begin{align*}
\pi_I &= \E[(v - p_0 - \lambda(x + u + \varepsilon)) \cdot \beta(v - p_0)] \\
      &= \beta\,\E[(v-p_0)^2] - \lambda\beta\,\E[(x+u+\varepsilon)(v-p_0)] \\
      &= \beta\,\sigv^2 - \tfrac{1}{2}\,\beta\,\sigv^2
       = \tfrac{1}{2}\,\beta\,\sigv^2.
\end{align*}
Substituting $\beta = \sqrt{\sigu^2+\sige^2}/\sigv$ yields
\eqref{eq:pi-I}.

\emph{(ii) Noise.} Since $u \indep (x, \varepsilon)$ and
$\E[u] = 0$,
\begin{align*}
\pi_N &= \E[(v - p_0 - \lambda(x+u+\varepsilon)) \cdot u]
       = -\lambda\,\E[u\,(x+u+\varepsilon)]
       = -\lambda\,\sigu^2.
\end{align*}
With $\lambda = \sigv/(2\sqrt{\sigu^2+\sige^2})$, this gives
\eqref{eq:pi-N}.

\emph{(iii) MM/protocol.} By the zero-sum identity,
\[
\pi_M = -(\pi_I + \pi_N)
      = -\tfrac{1}{2}\,\sigv\sqrt{\sigu^2+\sige^2}
        + \frac{\sigv\,\sigu^2}{2\sqrt{\sigu^2+\sige^2}}
      = -\frac{\sigv\,\sige^2}{2\sqrt{\sigu^2+\sige^2}},
\]
which is \eqref{eq:pi-M}. \qed
\end{proof}

\subsection{The privacy subsidy}

The main welfare result follows directly.

\begin{theorem}[Privacy subsidy]
\label{thm:welfare}
Under the equilibrium of \Cref{thm:equilibrium}, the protocol's
expected loss per period is
\begin{equation}
|\pi_M| \;=\; \frac{\sigv\,\sige^2}{2\sqrt{\sigu^2 + \sige^2}}
        \;\geq\; 0,
\label{eq:subsidy}
\end{equation}
with equality if and only if $\sige = 0$. The quantity $|\pi_M|$
is the \emph{privacy subsidy}: the per-period transfer from the
protocol/LP pool to traders, induced by the privacy
mechanism. For protocol break-even, the total fees collected per
period must satisfy
\[
\mathrm{fees}_{\text{period}} \;\geq\; |\pi_M|.
\]
\end{theorem}

\begin{proof}
Equation~\eqref{eq:subsidy} is \eqref{eq:pi-M} of \Cref{lem:pl};
non-negativity is immediate. Equality at $\sige = 0$ recovers
the textbook-Kyle MM zero-profit identity (\Cref{rem:sanity}). \qed
\end{proof}

\paragraph{Naming and fee model.}
We label $|\pi_M|$ \emph{privacy subsidy} as a memorable analog to
LVR; equivalent neutral phrasings are \emph{noise-induced LP loss}
or \emph{protocol's adverse-selection cost from privacy}. The
break-even bound is computed against the no-fee equilibrium: a
per-trade fee $f$ breaks the linear-strategy structure (informed
traders cease trading when $|v - p_0| < f$), so a complete
fee-equilibrium analysis with endogenous volume response is left
for future work.

\subsection{Comparative statics of the subsidy}

\begin{proposition}[Subsidy asymptotics and shape]
\label{prop:subsidy-asymptotics}
Let $|\pi_M|(\sige)$ denote the subsidy of \eqref{eq:subsidy} as a
function of $\sige \geq 0$ with $\sigv, \sigu$ fixed.
\begin{enumerate}[leftmargin=*]
\item \emph{Low-privacy expansion}: as $\sige \downarrow 0$,
$|\pi_M|(\sige) = \tfrac{\sigv}{2\sigu}\,\sige^2 + O(\sige^4)$.
\item \emph{High-privacy limit}: as $\sige \to \infty$,
$|\pi_M|(\sige) = \tfrac{1}{2}\,\sigv\,\sige + O(\sige^{-1})$.
\item $|\pi_M|$ is strictly increasing in $\sige$ on $[0, \infty)$.
\item $|\pi_M|$ has a single inflection point at $\sige^\star =
\sqrt{2}\,\sigu$: it is convex on $[0, \sige^\star]$ and concave on
$[\sige^\star, \infty)$.
\end{enumerate}
\end{proposition}

\begin{proof}[Proof of (iii)--(iv)]
Differentiating $|\pi_M|$ twice in $\sige$,
\[
\frac{\partial |\pi_M|}{\partial \sige}
= \frac{\sigv\,\sige\,(2\sigu^2 + \sige^2)}{2\,(\sigu^2 + \sige^2)^{3/2}}
> 0
\]
for $\sige > 0$, establishing (iii). The second derivative simplifies
to
\[
\frac{\partial^2 |\pi_M|}{\partial \sige^2}
= \frac{\sigv\,\sigu^2\,(2\sigu^2 - \sige^2)}{2\,(\sigu^2 + \sige^2)^{5/2}},
\]
which vanishes exactly at $\sige^2 = 2\sigu^2$, is positive for
$\sige^2 < 2\sigu^2$, and negative for $\sige^2 > 2\sigu^2$. This
gives (iv). \qed
\end{proof}

Item (i): for small privacy noise, the subsidy is quadratic in
$\sige$ --- doubling the privacy parameter quadruples the LP-pool
cost. Item (ii): asymptotically linear growth in the high-privacy
regime. Item (iv): the protocol's marginal cost
$\partial|\pi_M|/\partial\sige$ is itself increasing for small
$\sige$ but eventually decreasing as $\sige$ grows past
$\sqrt{2}\,\sigu$ --- past this threshold, additional privacy is
\emph{less} expensive at the margin. The inflection at
$\sige^\star = \sqrt{2}\,\sigu$ is a structural feature of the
square-root denominator in \eqref{eq:subsidy}.

\subsection{Welfare incidence}

\begin{corollary}[Noise traders also benefit from privacy]
\label{cor:noise-helped}
$\partial \pi_N / \partial \sige > 0$ for all $\sige > 0$: the
noise (uninformed) traders' expected loss is \emph{strictly
decreasing} in the privacy parameter.
\end{corollary}

\begin{proof}
Differentiating \eqref{eq:pi-N},
\[
\frac{\partial \pi_N}{\partial \sige}
= \frac{\sigv\,\sigu^2\,\sige}{2\,(\sigu^2 + \sige^2)^{3/2}} > 0
\quad\text{for } \sige > 0. \qed
\]
\end{proof}

Combined with \Cref{thm:welfare}, the welfare picture under privacy
is therefore:
\begin{itemize}[leftmargin=*]
\item informed trader gains: $\partial \pi_I / \partial \sige > 0$;
\item noise trader gains (\emph{loses less}): $\partial \pi_N /
\partial \sige > 0$;
\item protocol/LP pool loses by exactly that sum:
$\partial |\pi_M| / \partial \sige > 0$;
\item total welfare is preserved (zero sum).
\end{itemize}

Privacy improves $\pi_N$, but not at the informed trader's
expense. Both trader types gain; the entire transfer is borne
by the protocol. In the low-privacy
regime, the leading-order Taylor expansions of $\pi_I$ and $\pi_N$
around $\sige = 0$ are equal,
\[
\pi_I(\sige) - \pi_I(0)
\;=\; \pi_N(\sige) - \pi_N(0)
\;=\; \frac{\sigv}{4\sigu}\,\sige^2 + O(\sige^4),
\]
so the privacy subsidy is split symmetrically between informed and
noise traders to leading order. The asymmetry only appears at
higher orders: in the high-privacy limit, the informed trader
captures essentially all the subsidy ($\pi_I \sim \sigv\sige/2$)
while the noise trader's loss vanishes
($\pi_N \to 0$, i.e.\ their gain over the classical Kyle benchmark
saturates at $\sigv\sigu/2$).

\begin{remark}[The privacy "gain" is gross-of-fees;
\Cref{cor:noise-helped} is welfare-neutral net-of-fees]
\label{rem:welfare-neutral}
The statement ``noise traders also benefit from privacy''
(\Cref{cor:noise-helped}) is a \emph{gross-of-fees} observation
about the no-fee equilibrium of \Cref{thm:equilibrium}. We now
show that, at the same no-fee equilibrium volumes, charging the
break-even fee of \Cref{thm:welfare} via a volume-proportional
levy exactly cancels both traders' incremental gains.

At equilibrium the expected absolute volumes are
$\E|x| = (\sigv/2\lambda)\sqrt{2/\pi}$ and
$\E|u| = \sigu\sqrt{2/\pi}$, so total volume
$Q := \E|x| + \E|u| = \sqrt{2/\pi}\,(\sqrt{\sigu^2 + \sige^2} + \sigu)$
using $\lambda = \sigv/(2\sqrt{\sigu^2 + \sige^2})$ from
\Cref{thm:equilibrium}. The break-even rate is $f = |\pi_M|/Q$.
Each side pays its volume share:
\begin{align*}
  f \cdot \E|x|
  &\;=\; \frac{\sigv\,(\sqrt{\sigu^2+\sige^2} - \sigu)}{2}
  \;=\; \pi_I(\sige) - \pi_I(0), \\
  f \cdot \E|u|
  &\;=\; \frac{\sigv\,\sigu\,(\sqrt{\sigu^2+\sige^2} - \sigu)}
              {2\sqrt{\sigu^2+\sige^2}}
  \;=\; \pi_N(\sige) - \pi_N(0).
\end{align*}
Each fee equals the corresponding trader's incremental gain
over the $\sige = 0$ baseline; net-of-fees, both
$\pi_I - f\,\E|x| = \sigv\sigu/2$ and
$\pi_N - f\,\E|u| = -\sigv\sigu/2$ revert to the classical Kyle
values, while the MM is exactly compensated. Privacy is therefore
\emph{exactly welfare-neutral} under the volume-proportional
break-even fee, at the partial-equilibrium level of analysis
(no-fee equilibrium volumes, with fee revenue redistributed back
to the LP pool). The full fee-equilibrium analysis, in which the
fee distorts the linear-strategy structure as flagged after
\Cref{thm:welfare}, remains open and may yield additional
deadweight loss.
\end{remark}

\section{Application to zero-knowledge exchange designs}
\label{sec:application}

The model of \Cref{sec:setup} applies cleanly to one class of
zero-knowledge exchange designs: smart-contract AMMs whose
observation of order flow is perturbed by an explicit additive
Gaussian noise injection. We treat this canonical application in
\Cref{ssec:shielded-amm}, then situate three nearby privacy
designs that do \emph{not} fit our framework as currently
formulated in \Cref{ssec:outside}.

\subsection{Primary application: shielded AMMs with DP-style noise}
\label{ssec:shielded-amm}

A \emph{shielded AMM with differential-privacy (DP) noise} is a
constant-function-style AMM augmented with a privacy layer that
injects calibrated Gaussian noise into the observable order flow
before the AMM applies its update rule. Concretely, after each
trade the protocol observes $\ytil = y + \varepsilon$ with
$\varepsilon \sim N(0, \sige^2)$ drawn by the privacy mechanism,
and the AMM updates reserves --- and therefore the spot price ---
according to this noisy signal. The privacy parameter $\sige$ is
exactly the standard deviation of the injected DP noise.

\begin{sloppypar}
The mapping is literal: $\sige$ in our framework equals the DP
noise scale in the implementation. \Cref{thm:welfare} therefore
yields a direct break-even fee prescription for any such shielded
AMM: per-period fees must total at least the subsidy $|\pi_M|$ to
compensate the LP pool.
\ifwine
\else
\Cref{app:numerical} tabulates the resulting fee
floor in USD per day under a BTC/USDT calibration; at
$\sige = \sigu$, the LP/protocol must collect approximately
$\$1$M/day --- comparable to the entire revenue from a $0.1\%$
fee on $\$1$B of daily volume.
\fi
\end{sloppypar}

This setting pairs cleanly with the DP-feasibility analysis of
Chitra, Angeris, and Evans~\cite{chitra2022dpcfmm}, who study a
\emph{Uniform Random Execution} mechanism that achieves
$(\varepsilon,\delta)$-DP in constant-function market makers and
characterize the privacy parameter as a function of curvature and
trade count. Their analysis is on the realizability side: when
and how DP can be achieved in a CFMM. Our framework is the
complementary welfare side: given that a DP layer of intensity
$\sige$ is in place, what does the LP pool pay?

\subsection{Mechanisms outside our framework}
\label{ssec:outside}

Three nearby privacy designs warrant separate analysis frameworks
that our additive-Gaussian-noise model does not capture.

\paragraph{Penumbra-style batched swaps.}
Penumbra~\cite{penumbra-docs} clears swaps via per-block batches
in which the cleared price depends on the aggregate batch flow.
The Bayesian-MM observation is the \emph{exact} aggregate
$Y_\tau = \sum_t (x_t + u_t)$ over $\tau$ pooled noise-trader
draws --- not a noise-perturbed version of $Y_\tau$. Treating one
batch as one period (so per-batch noise-trader variance is
$\tau\sigu^2$), direct calculation gives $\lambda =
\sigv/(2\sigu\sqrt{\tau})$ and the informed batch-total
coefficient $\beta = \sigu\sqrt{\tau}/\sigv$ --- textbook Kyle
with the rescaling $\sigu \to \sigu\sqrt{\tau}$. Because the MM
both prices on and trades against the same observable $Y_\tau$,
the Bayesian projection identity gives $\pi_M = 0$ exactly.
Batching reshapes market depth and informed-trader intensity but
generates no privacy subsidy under this framework.

\paragraph{Suave-style sealed-bid order-flow auctions.}
Sealed bidding with delayed reveal~\cite{flashbots-suave}
creates \emph{temporal} information asymmetry rather than
additive observation noise. The block proposer commits to a
price ex ante and trades against revealed flow ex post; the
adverse-selection cost in this design is closer to LVR's
stale-price arbitrage (\Cref{sec:lvr}) than to our $\sige$-noise
setting. A continuous-time analysis with explicit time-lag is
required.

\paragraph{Renegade-style midpoint-pegged crossing.}
Renegade~\cite{renegade-wp} matches peer orders at an external
lit-exchange midpoint via MPC: the on-chain mechanism consumes
no flow signal to compute price. Such oracle-pegged designs fall
entirely outside any Kyle-style analysis, since price discovery
is exogenous.

\section{Connection to Loss-Versus-Rebalancing}
\label{sec:lvr}

The privacy subsidy of \Cref{thm:welfare} occupies the same
conceptual slot as the \emph{Loss-Versus-Rebalancing} (LVR) measure
of Milionis et al.~\cite{milionis2022lvr}. Both are closed-form,
per-period welfare quantities that measure the adverse-selection
cost borne by an automated pricing mechanism in the presence of
informed traders, and both yield direct break-even fee
prescriptions.

The source of the cost differs across the two results:
\begin{itemize}[leftmargin=*]
\item LVR captures the cost of \emph{stale prices}: the AMM's
price function is committed and lags the true reference price
which arbitrageurs can read externally. The information asymmetry
is temporal.
\item The privacy subsidy captures the cost of \emph{privacy
noise}: the AMM observes its price-relevant signal with additive
Gaussian noise injected by the privacy mechanism. The
information asymmetry is informational, not temporal.
\end{itemize}

\begin{sloppypar}
A natural open question is whether these costs combine additively
in a continuous-time AMM that is simultaneously price-lagged and
privacy-aggregated. A formal model would require lifting LVR's
continuous-time setup and our single-period Kyle setup into a
common framework, which we do not attempt here. As a heuristic
suggestion for follow-up work, we conjecture that the
break-even fee in such a setting decomposes to leading order as
\end{sloppypar}
\[
\mathrm{fees}_{\text{period}} \;\geq\; \mathrm{LVR}
   \,+\, |\pi_M|,
\]
but emphasize that this is speculative. The cross-terms between
time-lag and privacy noise need not vanish at higher orders: an
arbitrageur exploiting a stale price may also exploit the
privacy-noise statistics, and these channels' welfare costs need
not be independent. A vanishing of cross-terms at leading order
would require $\varepsilon$ to be temporally uncorrelated with the
stale-price process. A rigorous treatment is left for future work.

\section{Conclusion}
\label{sec:conclusion}

We have derived the closed-form linear Kyle equilibrium in a
single-period market in which a committed Bayesian AMM observes
order flow perturbed by independent Gaussian privacy noise of
scale $\sige$. The equilibrium price-impact coefficient is
$\lambda = \sigv/(2\sqrt{\sigu^2 + \sige^2})$, monotonically
decreasing in $\sige$. The corresponding informed-trader strategy
intensifies symmetrically as $\beta = \sqrt{\sigu^2+\sige^2}/\sigv$,
and the product $\lambda\beta = 1/2$ is invariant in $\sige$.
The welfare decomposition identifies a per-period transfer
$|\pi_M| = \sigv\sige^2/(2\sqrt{\sigu^2+\sige^2})$ from the
protocol/LP pool to traders. This is the
\emph{privacy subsidy}, the break-even fee that any
privacy-aggregated exchange must charge to compensate its
liquidity layer. The result applies directly to shielded
AMMs with additive-noise (differential-privacy) injection.
Related privacy designs --- Penumbra-style batched swaps,
Suave-style sealed-bid order-flow auctions, and oracle-pegged
crossings such as Renegade --- require separate frameworks:
batching reduces to Kyle with rescaled noise-trader variance
(no subsidy under Bayesian MM); sealed-bid auctions sit closer
to LVR's temporal-asymmetry regime; oracle-pegged designs fall
entirely outside Kyle.

\paragraph{Future work.}
Extensions left to subsequent work include a Glosten--Milgrom
bid-ask-spread analog under privacy noise; other privacy
mechanisms (directional-only, bucketed, time-delayed observation);
multi-period dynamics aimed at a combined LVR-plus-privacy
decomposition; mechanized formalization in a proof assistant;
and empirical calibration of $\sige$ from live zero-knowledge
exchange traces.

\appendix

\section{Numerical illustration}
\label{app:numerical}

The closed-form results of \Cref{thm:equilibrium,thm:welfare}
admit immediate numerical evaluation. \Cref{fig:subsidy} plots
$|\pi_M|(\sige)$ for $\sigv = \sigu = 1$ over $\sige \in [0, 5]$;
\Cref{tab:dimensionless,tab:btc} tabulate $\lambda$, $\beta$, and
$|\pi_M|$ at representative parameter values.

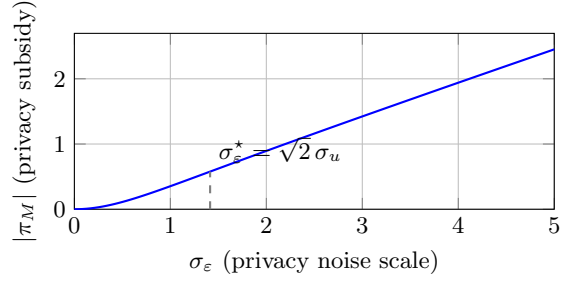
\begin{figure}[H]
\centering
\begin{tikzpicture}
\begin{axis}[
    xlabel = {$\sigma_\varepsilon$ (privacy noise scale)},
    ylabel = {$|\pi_M|$ (privacy subsidy)},
    width  = 0.65\textwidth,
    height = 0.32\textwidth,
    domain = 0:5,
    samples = 120,
    grid    = both,
    xmin = 0, xmax = 5,
    ymin = 0
]
\addplot[thick, blue] {x^2 / (2*sqrt(1 + x^2))};
\addplot[dashed, gray, thick] coordinates {(1.4142, 0) (1.4142, 0.577)};
\node[above right, font=\small] at (axis cs:1.4142, 0.577)
    {$\sigma_\varepsilon^\star = \sqrt{2}\,\sigma_u$};
\end{axis}
\end{tikzpicture}
\caption{Privacy subsidy $|\pi_M|$ vs. $\sige$ for
$\sigv = \sigu = 1$. The convexity-to-concavity inflection at
$\sige^\star = \sqrt{2}$ marks the transition between the quadratic
low-privacy regime and the linear high-privacy regime.}
\label{fig:subsidy}
\end{figure}

\subsection{Dimensionless table}

\Cref{tab:dimensionless} reports $\lambda$, $\beta$, and $|\pi_M|$
in units where $\sigv = \sigu = 1$. The convexity-to-concavity
inflection of $|\pi_M|$ at $\sige^\star = \sqrt{2}\,\sigu \approx
1.414$ (\Cref{prop:subsidy-asymptotics}(iv)) is bracketed by the
$\sige \in \{1.0, 1.414, 2.0\}$ rows.

\begin{table}[H]
\centering
\begin{tabular}{c|ccc|c}
$\sige$ & $\lambda$ & $\beta$ & $|\pi_M|$ & note \\
\hline
$0$    & $0.500$ & $1.000$ & $0.000$ & textbook Kyle \\
$0.5$  & $0.447$ & $1.118$ & $0.112$ & low-privacy regime \\
$1.0$  & $0.354$ & $1.414$ & $0.354$ & $\sige = \sigu$ \\
$\sqrt{2}$ & $0.289$ & $1.732$ & $0.577$ & $\sige = \sige^\star$ (inflection) \\
$2.0$  & $0.224$ & $2.236$ & $0.894$ & past inflection \\
$3.0$  & $0.158$ & $3.162$ & $1.423$ & high-privacy \\
$5.0$  & $0.098$ & $5.099$ & $2.451$ & far high-privacy \\
\end{tabular}
\caption{Dimensionless equilibrium values for
$\sigv = \sigu = 1$, varying $\sige$. The privacy subsidy
$|\pi_M|$ grows quadratically near $\sige = 0$ (slope $\to 0$),
becomes maximally steep near $\sige^\star = \sqrt{2}\,\sigu$, and
asymptotes to linear $|\pi_M| \approx \sige/2$ for large $\sige$.}
\label{tab:dimensionless}
\end{table}

\subsection{BTC-calibrated example}

For an illustrative BTC/USDT calibration, take a per-day window
with $\sigv = \$3{,}000$ (corresponding to $\sim 3\%$ daily
volatility on a $\$100{,}000$ asset) and $\sigu = 1{,}000$~BTC
per day. \Cref{tab:btc} reports $|\pi_M|$ in USD per day across
several $\sige$ values, expressed as a multiple of $\sigu$.

\begin{table}[H]
\centering
\begin{tabular}{c|c|c}
$\sige / \sigu$ & $|\pi_M|$ (USD/day) & As fraction of $\sigv \sigu$ \\
\hline
$0.1$  & \$$\sim 15{,}000$ & $0.0050$ \\
$0.5$  & \$$\sim 335{,}000$ & $0.112$ \\
$1.0$  & \$$\sim 1{,}060{,}000$ & $0.354$ \\
$\sqrt{2}$ & \$$\sim 1{,}730{,}000$ & $0.577$ \\
$2.0$  & \$$\sim 2{,}680{,}000$ & $0.894$ \\
\end{tabular}
\caption{Per-day privacy subsidy in USD for a BTC/USDT
calibration with $\sigv = \$3{,}000$/day and $\sigu = 1{,}000$~BTC/day.
($\sigu$ is the noise-trader-component standard deviation, not
total daily volume; typical BTC/USDT daily volume is one to two
orders of magnitude larger.) At $\sige = \sigu$, the LP/protocol
must collect approximately $\$1$M/day in fees to break even.}
\label{tab:btc}
\end{table}

For context, BTC/USDT spot volume on a typical centralized exchange
is on the order of $\$1$--$\$5$~billion per day. A 0.1\% fee on
$\$1$B daily volume yields $\$1$M/day in revenue. Under the
calibration of \Cref{tab:btc}, this revenue roughly breaks even
against the privacy subsidy at $\sige = \sigu$ (the subsidy is
$\sim\$1.06$M, so a $0.1\%$ fee on $\$1$B falls about $6\%$ short),
and falls further short for larger $\sige$. The implication for privacy-aggregated exchange
designers is that fee schedules and privacy parameters cannot be
chosen independently.

\bibliographystyle{splncs04}
\bibliography{references}

\end{document}